\documentclass[copyright,creativecommons,noncommercial]{eptcs}
\providecommand{\event}{QPL 2012}
\usepackage{breakurl,amsmath,amsthm,amssymb,xspace,xypic,enumerate,color,tikzfig}

\tikzstyle{every picture}=[baseline=-0.25em,shorten <=-0.1pt]
\tikzstyle{dotpic}=[scale=0.5]
\tikzstyle{braceedge}=[decorate,decoration={brace,amplitude=1mm,raise=-1mm}]
\tikzstyle{dot}=[inner sep=0.7mm,minimum width=0pt,minimum height=0pt,fill=black,draw=black,shape=circle]
\tikzstyle{black dot}=[dot]
\tikzstyle{white dot}=[dot,fill=white]
\tikzstyle{gray dot}=[dot,fill=gray!40!white]
\tikzstyle{alt white dot}=[white dot,label={[xshift=3mm,yshift=-0.05mm,font=\tiny]left:$*$}]
\tikzstyle{alt gray dot}=[gray dot,label={[xshift=3mm,yshift=-0.05mm,font=\tiny]left:$*$}]
\tikzstyle{white norm}=[rectangle,fill=white,draw=black,minimum height=2mm,minimum width=2mm,inner sep=0pt,font=\small]
\tikzstyle{gray norm}=[white norm,fill=gray!40!white]
\tikzstyle{square box}=[rectangle,fill=white,draw=black,minimum height=5mm,minimum width=5mm,font=\small]
\tikzstyle{square gray box}=[rectangle,fill=gray!30,draw=black,minimum height=6mm,minimum width=6mm]
\tikzstyle{diredge}=[->]
\tikzstyle{rdiredge}=[<-]
\tikzstyle{dashed edge}=[dashed]
\tikzstyle{cross}=[preaction={draw=white, -, line width=3pt}]

\newcommand{\pantsalg}{%
\,\begin{tikzpicture}[dotpic,scale=0.8]
    \node [style=none] (0) at (-0.75, -0.5) {};
    \node [style=none] (1) at (0.25, 0.5) {};
    \node [style=none] (2) at (0.75, -0.5) {};
    \node [style=none] (3) at (0.25, -0.5) {};
    \node [style=none] (4) at (-0.25, 0.5) {};
    \node [style=none] (5) at (-0.25, -0.5) {};
    \draw [style=diredge, in=-90, out=90] (2.center) to (1.center);
    \draw [style=diredge, in=90, out=-90] (4.center) to (0.center);
    \draw [style=diredge, in=90, out=90, looseness=2.00] (5.center) to (3.center);
\end{tikzpicture}\,}
\newcommand{\dotcounit}[1]{%
\,\begin{tikzpicture}[dotpic,yshift=-1mm]
\node [#1] (a) at (0,0.25) {}; 
\draw [diredge] (0,-0.2)--(a);
\end{tikzpicture}\,}
\newcommand{\dotunit}[1]{%
\,\begin{tikzpicture}[dotpic,yshift=1.5mm]
\node [#1] (a) at (0,-0.25) {}; 
\draw [diredge] (a)--(0,0.2);
\end{tikzpicture}\,}
\newcommand{\dotcomult}[1]{%
\,\begin{tikzpicture}[dotpic,yshift=0.5mm]
  \node [#1] (a) {};
  \draw [diredge] (-90:0.55)--(a);
  \draw [diredge] (a) -- (45:0.6);
  \draw [diredge] (a) -- (135:0.6);
\end{tikzpicture}\,}
\newcommand{\dotmult}[1]{%
\,\begin{tikzpicture}[dotpic]
  \node [#1] (a) {};
  \draw [diredge] (a) -- (90:0.55);
  \draw [diredge] (a) (-45:0.6) -- (a);
  \draw [diredge] (a) (-135:0.6) -- (a);
\end{tikzpicture}\,}
\newcommand{\dotaction}[1]{%
\,\begin{tikzpicture}[dotpic,yshift=0.5mm]
  \node [#1] (a) {};
  \draw [diredge] (-90:0.55)--(a);
  \draw [diredge] (a) -- (45:0.6);
  \draw [rdiredge] (a) -- (135:0.6);
\end{tikzpicture}\,}
\newcommand{\dotcoaction}[1]{%
\,\begin{tikzpicture}[dotpic]
  \node [#1] (a) {};
  \draw [diredge] (a) -- (90:0.55);
  \draw [diredge] (a) (-45:0.6) -- (a);
  \draw [rdiredge] (a) (-135:0.6) -- (a);
\end{tikzpicture}\,}
\newcommand{\dotdualmult}[1]{%
\!\begin{tikzpicture}[dotpic]
    \node [style=white dot] (0) at (0, 0.3) {};
    \node [style=none] (1) at (-0.5, -0.4) {};
    \node [style=none] (2) at (0.5, -0.4) {};
    \node [style=none] (3) at (0, 0.8) {};
    \draw [style=diredge] (3.center) to (0);
    \draw [style=diredge, in=15, out=-30, looseness=1.50] (0) to (1.center);
    \draw [style=diredge, in=165, out=-150, looseness=1.50] (0) to (2.center);
\end{tikzpicture}\!}

\newcommand{\dotnorm}[1]{%
\,\begin{tikzpicture}[dotpic,yshift=0.4mm]
    \node [style=none] (0) at (0, -0.4) {};
    \node [style=#1] (1) at (0, -0) {};
    \node [style=none] (2) at (0, 0.5) {};
    \draw (0.center) to (1);
    \draw [style=diredge] (1) to (2.center);
\end{tikzpicture}\,}
\newcommand{\dotconorm}[1]{%
\,\begin{tikzpicture}[dotpic,yshift=0.4mm]
    \node [style=none] (0) at (0, -0.4) {};
    \node [style=white norm] (1) at (0, 0.1) {};
    \node [style=none] (2) at (0, 0.5) {};
    \draw [style=diredge] (1) to (0.center);
    \draw (2.center) to (1);
\end{tikzpicture}\,}

\newcommand{\blackunit}{\dotunit{dot}}

\newcommand{\blackmult}{\dotmult{dot}}
\newcommand{\blackcomult}{\dotcomult{dot}}
\newcommand{\blackaction}{\dotaction{dot}}

\newcommand{\whiteunit}{\dotunit{white dot}}
\newcommand{\whitecounit}{\dotcounit{white dot}}
\newcommand{\whitemult}{\dotmult{white dot}}
\newcommand{\whitecomult}{\dotcomult{white dot}}

\newcommand{\whitenorm}{\dotnorm{white norm}}
\newcommand{\whiteconorm}{\dotconorm{white norm}}
\newcommand{\whiteaction}{\dotaction{white dot}}
\newcommand{\whitecoaction}{\dotcoaction{white dot}}
\newcommand{\whitedualmult}{\dotdualmult{white dot}}

\newcommand{\graycounit}{\dotcounit{gray dot}}
\newcommand{\graymult}{\dotmult{gray dot}}

\newcommand{\graynorm}{\dotnorm{gray norm}}

\newcommand{\grayaction}{\dotaction{gray dot}}
\newcommand{\graycoaction}{\dotcoaction{gray dot}}
\newcommand{\prodmult}[2]{\dotmult{#1}\!\!\!\!\dotmult{#2}}

\newcommand{\cat}[1]{\ensuremath{\mathbf{#1}}\xspace}
\newcommand{\FHilb}{\cat{FHilb}}
\newcommand{\Rel}{\cat{Rel}}
\newcommand{\Stoch}{\cat{Stoch}}
\newcommand{\CPs}{\ensuremath{\mathrm{CP}^*}\xspace}
\newcommand{\CPM}{\ensuremath{\mathrm{CPM}}\xspace}
\newcommand{\Split}{\ensuremath{\mathrm{Split}^\dag}\xspace}
\newcommand{\V}{\cat{V}}
\newcommand{\C}{\ensuremath{\mathbb{C}}}
\newcommand{\inprod}[2]{\ensuremath{\langle #1\,|\,#2 \rangle}}
\DeclareMathOperator{\Tr}{Tr}
\newcommand{\id}[1][]{\ensuremath{1_{#1}}}
\DeclareMathOperator{\Mor}{Mor}
\DeclareMathOperator{\dom}{dom}
\DeclareMathOperator{\cod}{cod}
\newcommand{\bra}[1]{\ensuremath{\langle #1 |}}
\newcommand{\ket}[1]{\ensuremath{| #1 \rangle}}

\theoremstyle{plain}
\newtheorem{theorem}{Theorem}[section]
\newtheorem{lemma}[theorem]{Lemma}
\newtheorem{proposition}[theorem]{Proposition}
\theoremstyle{definition}
\newtheorem{definition}[theorem]{Definition}

\newtheorem{remark}[theorem]{Remark}

\title{Categories of Quantum and Classical Channels \\{\Large (extended abstract)}}
\author{Bob Coecke\thanks{Supported by the John Templeton Foundation.} \qquad Chris Heunen\thanks{Supported by EPSRC Fellowship EP/L002388/1.} \qquad Aleks Kissinger\footnotemark[1]
  \institute{University of Oxford, Department of Computer Science} 
  \email{\{coecke,heunen,alek\}@cs.ox.ac.uk} }
\date{}
\def\titlerunning{Categories of Quantum and Classical Channels}
\def\authorrunning{B. Coecke, C. Heunen \& A. Kissinger}

\begin{document}
\maketitle
\begin{abstract}
  We introduce the CP*--construction on a dagger compact closed
  category as a generalisation of Selinger's CPM--construction. While
  the latter takes a dagger compact closed category and forms its
  category of ``abstract matrix algebras'' and completely positive
  maps, the CP*--construction forms its category of ``abstract
  C*-algebras'' and completely positive maps. This analogy is
  justified 
  by the case of finite-dimensional Hilbert spaces, where the
  CP*--construction yields the category of finite-dimensional
  C*-algebras and completely positive maps. 
  
  The CP*--construction fully embeds Selinger's CPM--construction in
  such a way that the objects in the image of the embedding can be thought
  of as ``purely quantum'' state spaces. It also embeds the category
  of classical stochastic maps, whose image consists of ``purely
  classical'' state spaces. By allowing classical and quantum data to
  coexist, this provides elegant abstract notions of preparation,
  measurement, and more general quantum channels.
\end{abstract}

\section{Introduction}

One of the motivations driving categorical treatments of quantum
mechanics is to place classical and quantum systems on an equal
footing in a single category, so that one can study their interactions.
The main idea of categorical quantum
mechanics~\cite{abramskycoecke:cqm} is to fix a category (usually 
dagger compact) whose objects are thought of as state
spaces and whose morphisms are evolutions. There are two main
variations. 
\begin{itemize}
  \item ``Dirac style'': Objects form \textit{pure} state spaces, and
    isometric morphisms form \textit{pure}
    state evolutions. 
  \item ``von Neumann style'': Objects form spaces of \textit{mixed} states,
    and morphisms form \textit{mixed} quantum evolutions, also
    known as \emph{quantum channels}.
\end{itemize}

The prototypical example of a ``Dirac style'' category is $\FHilb$,
the category of finite-dimensional Hilbert spaces and linear maps. One
can pass to a ``von Neumann style category'' by considering the
C*-algebras of operators on these Hilbert spaces, with completely
positive maps between them.
Selinger's CPM--construction provides an abstract bridge from the
former to the latter~\cite{selinger:completelypositive}, turning
morphisms in $\V$ into quantum channels in $\CPM[\V]$. 
However, this passage loses the connection between quantum and classical
channels. For example, $\CPM[\FHilb]$ only includes objects corresponding
to the entire state space of some quantum system, whereas we would
often like to focus in subspaces corresponding to particular classical
contexts. There have been several proposals to remedy this situation, which
typically involve augmenting $\CPM[\V]$ with extra objects to carry this
classical structure. Section~\ref{sec:related} provides an overview.

This extended abstract\footnote{A long version of this extended abstract is
now available:~\cite{journalversion}.} introduces a new solution to this problem, called the
CP*--construction, which is closer in spirit to the study of quantum
information using C*-algebras (see e.g.~\cite{keyl:quantuminformation}). Rather
than freely augmenting a category of quantum data with classical objects, we
define a new category whose objects are ``abstract C*--algebras'', and whose
morphisms are the analogue of completely positive maps. It is then possible to
construct a full embedding of the category $\CPM[\V]$, whose image yields the
purely quantum channels. Furthermore, there exists another full embedding of
the category $\Stoch[\V]$ of classical stochastic maps, whose image yields the
purely classical channels. The remainder of the category $\CPs[\V]$ can be
interpreted as mixed classical/quantum state spaces, which carry partially
coherent quantum states, such as degenerate quantum measurement outcomes.

\[\xymatrix@C+8ex@R-6ex{
  \V \ar[r]
  & \CPM[\V] \ar[r]^-{\text{\scriptsize{full, faithful}}} 
  & \CPs[\V]
  & \Stoch[\V] \ar[l]_-{\text{\scriptsize{full, faithful}}}
}\]


The paper is structured as follows. Section~\ref{sec:FA} introduces
normalisable dagger Frobenius algebras. These form a crucial ingredient to the
CP*--construction, which is defined in Section~\ref{sec:CPs}.
Sections~\ref{sec:embedding} and~\ref{sec:stoch} then show that the
CP*--construction simultaneously generalises Selinger's CPM--construction,
consisting of ``quantum channels'', and the \Stoch--construction, consisting
of ``classical channels''. We then consider two examples:
Section~\ref{sec:Hilb} shows that $\CPs[\FHilb]$ is the category of finite-
dimensional C*-algebras and completely positive maps, and
Section~\ref{sec:Rel} shows that $\CPs[\Rel]$ is the category of (small)
groupoids and inverse-respecting relations. Section~\ref{sec:idempotents}
compares $\CPs[\V]$ with Selinger's extension of the CPM--construction using
split idempotents. Finally, Section~\ref{sec:future} discusses the many
possibilities this opens up. 

\subsection{Related Work}\label{sec:related}

Selinger introduced two approaches to add classical data to $\CPM[\V]$ by
either freely adding biproducts to $\CPM[\V]$ or freely splitting the 
$\dagger$-idempotents of $\CPM[\V]$~\cite{selinger:idempotents}. These
new categories are referred to as $\CPM[\V]^{\oplus}$ and
$\Split[\CPM[\V]]$, respectively.

When $\V = \FHilb$, both categories provide ``enough space'' to reason about
classical and quantum data, as any finite-dimensional C*--algebra can be
defined as a sum of matrix algebras (as in $\CPM[\FHilb]^{\oplus}$) or as a
certain orthogonal subspace of a larger matrix algebra (as in
$\Split[\CPM[\FHilb]]$). However, it is unclear whether the second
construction captures too much: its may contain many more objects than simply
mixtures of classical and quantum state
spaces~\cite[Remark~4.9]{selinger:idempotents}. On the other hand, when $\V
\neq \FHilb$, the category $\CPM[\FHilb]^{\oplus}$ may be too
small. That is, there may be interesting objects that are not just
sums of quantum objects. 

For this reason, it is interesting to study $\CPs[\V]$, as it lies between these
two constructions:
\[\xymatrix@C+8ex@R-6ex{
    \CPM[\V]^{\oplus} \ar[r]^-{\text{\scriptsize{full, faithful}}}
  & \CPs[\V] \ar[r]^-{\text{\scriptsize{full, faithful}}}
  & \Split[\CPM[\V]]
}\]

The first embedding is well-defined whenever \V has biproducts, and the second
when \V satisfies a technical axiom about square roots (see
Definition~\ref{def:alg-sq-roots}). In the former case, $\CPs[\V]$ inherits
biproducts from \V, so it is possible to lift the embedding of $\CPM[\V]$ by
the universal property of the free biproduct completion; this will be
proved in detail in a subsequent paper. In the latter case, one
can always construct the associated dagger-idempotent of an object in
$\CPs[\V]$, and (with the assumption from Definition~\ref{def:alg-sq-roots}),
the notions of complete positivity in $\CPs[\V]$ and $\Split[\CPM[\V]]$
coincide. We provide the details of this construction in
Section~\ref{sec:idempotents}.

A third approach by Coecke, Paquette, and Pavlovic is similar to ours in that
it makes use of commutative Frobenius algebras to represent classical
data~\cite{coeckepaquettepavlovic:structuralism}.
As in the previous two approaches, it takes $\CPM[\V]$ and freely adds
classical structure, this time by forming the comonad associated with a
particular commutative Frobenius algebra and taking the Kleisli
category. All such categories are then glued together by Grothendieck
construction. The CP*--construction was originally conceived as a way to
simplify this approach and overcome is limitations.

\section{Abstract C*-algebras}\label{sec:FA}

This section defines so-called normalisable dagger Frobenius algebras,
which will play a central role in the CP*-construction of the next
section. We start by recalling the notion of a Frobenius algebra.
For an introduction to dagger (compact) categories and their graphical
calculus, we refer to~\cite{abramskycoecke:cqm,selinger:graphical}.

\begin{definition}
  A \emph{Frobenius algebra} is an object $A$ in a monoidal category together
  with morphisms depicted as $\whitemult$, $\whiteunit$, $\whitecomult$
  and $\whitecounit$ on it satisfying the following diagrammatic equations.
  \ctikzfig{fa_axioms}
\end{definition}

Any Frobenius algebra defines a cap and a cup that satisfy the snake equation.
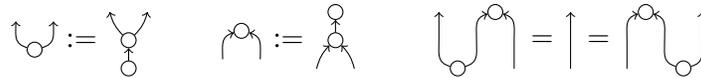
\ctikzfig{cap_cup}


\begin{definition}
  A Frobenius algebra is \emph{symmetric} when $\;%
\beginpgfgraphicnamed{symmetric_1}
\begin{tikzpicture}[dotpic]
	\begin{pgfonlayer}{nodelayer}
		\node [style=white dot] (0) at (-4.5, -0.75) {};
		\node [style=none] (1) at (-5, -0.25) {};
		\node [style=none] (2) at (-4, -0.25) {};
		\node [style=none] (3) at (-4, 0.75) {};
		\node [style=none] (4) at (-5, 0.75) {};
		\node [style=none] (5) at (-2, 0.75) {};
		\node [style=white dot] (6) at (-1.5, -0.75) {};
		\node [style=none] (7) at (-1, -0) {};
		\node [style=none] (8) at (-1, 0.75) {};
		\node [style=none] (9) at (-2, -0) {};
		\node [style=none] (10) at (-3, -0) {$=$};
	\end{pgfonlayer}
	\begin{pgfonlayer}{edgelayer}
		\draw [in=-90, out=154] (0) to (1.center);
		\draw [style=diredge, in=-90, out=90] (1.center) to (3.center);
		\draw [in=-90, out=30] (0) to (2.center);
		\draw [style=diredge, in=-90, out=90] (2.center) to (4.center);
		\draw [in=-90, out=154] (6) to (9.center);
		\draw [style=diredge, in=-90, out=90, looseness=0.75] (9.center) to (5.center);
		\draw [in=-90, out=30] (6) to (7.center);
		\draw [style=diredge, in=-90, out=90, looseness=0.75] (7.center) to (8.center);
	\end{pgfonlayer}
\end{tikzpicture}}
\endpgfgraphicnamed\;$
  and $\;%
\beginpgfgraphicnamed{symmetric_2}
\begin{tikzpicture}[dotpic]
	\begin{pgfonlayer}{nodelayer}
		\node [style=white dot] (0) at (2.25, 0.75) {};
		\node [style=none] (1) at (1.75, -0.75) {};
		\node [style=none] (2) at (1.75, 0.25) {};
		\node [style=none] (3) at (2.75, 0.25) {};
		\node [style=none] (4) at (3.75, -0) {$=$};
		\node [style=none] (5) at (5.75, -0) {};
		\node [style=none] (6) at (4.75, -0) {};
		\node [style=none] (7) at (5.75, -0.75) {};
		\node [style=none] (8) at (2.75, -0.75) {};
		\node [style=white dot] (9) at (5.25, 0.75) {};
		\node [style=none] (10) at (4.75, -0.75) {};
	\end{pgfonlayer}
	\begin{pgfonlayer}{edgelayer}
		\draw [style=rdiredge, in=90, out=-154] (0) to (2.center);
		\draw [in=90, out=-90] (2.center) to (8.center);
		\draw [style=rdiredge, in=90, out=-30] (0) to (3.center);
		\draw [in=90, out=-90] (3.center) to (1.center);
		\draw [style=rdiredge, in=90, out=-154] (9) to (6.center);
		\draw [in=90, out=-90, looseness=0.75] (6.center) to (10.center);
		\draw [style=rdiredge, in=90, out=-30] (9) to (5.center);
		\draw [in=90, out=-90, looseness=0.75] (5.center) to (7.center);
	\end{pgfonlayer}
\end{tikzpicture}}
\endpgfgraphicnamed\;$.
\end{definition}


In a dagger compact category $\cat{V}$, a \emph{dagger Frobenius}
algebra is a Frobenius algebra that satisfies $\whitecomult = (\whitemult)^\dag$ and
$\whitecounit = (\whiteunit)^\dag$. Their import for us starts with
the following theorem.

\begin{theorem}[\cite{vicary:quantumalgebras}]\label{thm:dfa-cstar}
  Given a dagger Frobenius algebra $(A,\whitemult)$ in \FHilb, the
  following operation gives $A$ the structure of a C*-algebra.
  \ctikzfig{star_operation}
  Furthermore, all finite-dimensional C*-algebras arise this way.
  \qed
\end{theorem}

In light of this theorem, one might be tempted to consider
dagger Frobenius algebras to be the ``correct'' way to define the
abstract analogue of finite-dimensional C*-algebras. However, there
is one more condition, called \textit{normalisability}, that is
satisfied by \textit{all} dagger Frobenius algebras in \FHilb, yet
not by dagger Frobenius algebras in general. Before we come to that,
we introduce the notion of a central map for a monoid.

\begin{definition}
  A map $z : A \rightarrow A$ is \textit{central} for a monoid when
  $\whitemult \circ (z \otimes 1_A) = z \circ \whitemult = \whitemult \circ (1_A \otimes z)$.
\end{definition}

We call such a map central, because it corresponds uniquely to a point $p_z : I \rightarrow A$ in the centre of the monoid via left (or equivalently right) multiplication by $p_z$.
  
Recall that a map $g :A \rightarrow A$ in a dagger category is
\textit{positive} when $g = h^\dag \circ h$ for some map $h$. It is
called \textit{positive definite} when it is a positive
isomorphism. Using these conditions, we give the definition of a 
normalisability as a well-behavedness property of the ``loop'' map
$\whitemult \circ \whitecomult$. 

\begin{definition}\label{def:dnfa}
  A dagger Frobenius algebra is \emph{normalisable} when there is a central,
  positive-definite map $z$, called the \emph{normaliser} and depicted
  as $\whitenorm$, satisfying the following diagrammatic equation.
  \ctikzfig{normalisable}
\end{definition}

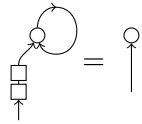
\emph{Special} dagger Frobenius algebras are normalisable dagger
Frobenius algebra where $z^2 = 1_A$. 
Normalisable dagger Frobenius algebras are always symmetric.  

\begin{theorem}
  Normalisable dagger Frobenius algebras are symmetric.
\end{theorem}
\begin{proof}
  The proof follows from expanding the counit and applying cyclicity of the trace $(*)$.
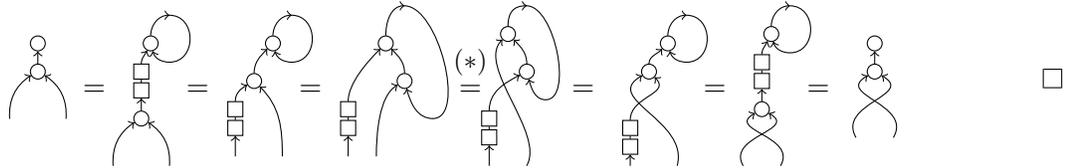
  \[%
\beginpgfgraphicnamed{norm_symmetric_pf}
\begin{tikzpicture}[dotpic]
	\begin{pgfonlayer}{nodelayer}
		\node [style=white dot] (0) at (-6.75, 1) {};
		\node [style=white norm] (1) at (-7, 0.25) {};
		\node [style=white norm] (2) at (-7, -0.25) {};
		\node [style=none] (3) at (-6.25, 0.5) {};
		\node [style=none] (4) at (-6.25, 1.75) {};
		\node [style=none] (5) at (-8.25, -0.25) {$=$};
		\node [style=white dot] (6) at (-9.75, 1) {};
		\node [style=white dot] (7) at (-9.75, 0.25) {};
		\node [style=none] (8) at (-10.5, -1) {};
		\node [style=none] (9) at (-9, -1) {};
		\node [style=white dot] (10) at (-7, -1) {};
		\node [style=none] (11) at (-7.75, -2.25) {};
		\node [style=none] (12) at (-6.25, -2.25) {};
		\node [style=none] (13) at (-3.25, -2) {};
		\node [style=none] (14) at (-3, 1.75) {};
		\node [style=none] (15) at (-4.5, -2) {};
		\node [style=none] (16) at (-3, 0.5) {};
		\node [style=none] (17) at (-5.5, -0.25) {$=$};
		\node [style=white dot] (18) at (-3.5, 1) {};
		\node [style=white dot] (19) at (-4, 0) {};
		\node [style=white norm] (20) at (-4.5, -0.75) {};
		\node [style=white norm] (21) at (-4.5, -1.25) {};
		\node [style=none] (22) at (0.75, -1) {};
		\node [style=none] (23) at (-0.75, -2) {};
		\node [style=none] (24) at (-2.5, -0.25) {$=$};
		\node [style=white norm] (25) at (-1.5, -0.75) {};
		\node [style=white norm] (26) at (-1.5, -1.25) {};
		\node [style=none] (27) at (0, 2) {};
		\node [style=white dot] (28) at (0, 0) {};
		\node [style=white dot] (29) at (-0.5, 1) {};
		\node [style=none] (30) at (-1.5, -2) {};
		\node [style=white dot] (31) at (2.75, 1.25) {};
		\node [style=none] (32) at (3.25, -2.25) {};
		\node [style=white norm] (33) at (2.25, -1.5) {};
		\node [style=white dot] (34) at (3.25, 0.25) {};
		\node [style=none] (35) at (1.75, -0.25) {$=$};
		\node [style=none] (36) at (3.25, 2) {};
		\node [style=none] (37) at (2.25, -2.25) {};
		\node [style=white norm] (38) at (2.25, -1) {};
		\node [style=none] (39) at (3.75, -0.5) {};
		\node [style=none] (40) at (7.5, 0.5) {};
		\node [style=none] (41) at (7.25, -2.25) {};
		\node [style=none] (42) at (4.75, -0.25) {$=$};
		\node [style=none] (43) at (7.5, 1.75) {};
		\node [style=white norm] (44) at (6, -1.25) {};
		\node [style=white norm] (45) at (6, -1.75) {};
		\node [style=white dot] (46) at (6.5, 0.25) {};
		\node [style=none] (47) at (8.25, -0.25) {$=$};
		\node [style=white dot] (48) at (7, 1) {};
		\node [style=none] (49) at (6, -2.25) {};
		\node [style=none] (50) at (11, -0.25) {$=$};
		\node [style=none] (51) at (10, -2.25) {};
		\node [style=none] (52) at (10.25, 2) {};
		\node [style=none] (53) at (9, -2.25) {};
		\node [style=none] (54) at (10.25, 0.75) {};
		\node [style=white dot] (55) at (9.75, 1.25) {};
		\node [style=white dot] (56) at (9.5, -0.75) {};
		\node [style=white norm] (57) at (9.5, 0.5) {};
		\node [style=white norm] (58) at (9.5, 0) {};
		\node [style=white dot] (59) at (12.5, 0.25) {};
		\node [style=none] (60) at (12, -1.5) {};
		\node [style=none] (61) at (13, -1.5) {};
		\node [style=white dot] (62) at (12.5, 1) {};
		\node [style=none] (63) at (1.75, 0.5) {\small $(*)$};
	\end{pgfonlayer}
	\begin{pgfonlayer}{edgelayer}
		\draw [style=diredge, in=180, out=90] (0) to (4.center);
		\draw [style=diredge, in=-117, out=90] (1) to (0);
		\draw (2) to (1);
		\draw [style=diredge, in=-75, out=180] (3.center) to (0);
		\draw [in=0, out=0, looseness=1.50] (4.center) to (3.center);
		\draw [style=diredge] (7) to (6);
		\draw [style=diredge, bend right] (9.center) to (7);
		\draw [style=diredge, bend left] (8.center) to (7);
		\draw [style=diredge, bend right] (12.center) to (10);
		\draw [style=diredge, bend left] (11.center) to (10);
		\draw [style=diredge] (10) to (2);
		\draw [style=diredge, in=180, out=90] (18) to (14.center);
		\draw (21) to (20);
		\draw [style=diredge, in=-75, out=180] (16.center) to (18);
		\draw [in=0, out=0, looseness=1.50] (14.center) to (16.center);
		\draw [style=diredge, in=-45, out=90] (13.center) to (19);
		\draw [style=diredge, in=-135, out=90] (19) to (18);
		\draw [style=diredge, in=-150, out=90] (20) to (19);
		\draw [style=diredge] (15.center) to (21);
		\draw [style=diredge, in=180, out=90] (29) to (27.center);
		\draw (26) to (25);
		\draw [in=0, out=0, looseness=0.75] (27.center) to (22.center);
		\draw [style=diredge] (30.center) to (26);
		\draw [style=diredge, in=-75, out=180] (22.center) to (28);
		\draw [style=diredge, in=-45, out=90] (28) to (29);
		\draw [style=diredge, in=-135, out=90, looseness=0.75] (25) to (29);
		\draw [style=diredge, in=-150, out=90, looseness=0.75] (23.center) to (28);
		\draw [style=diredge, in=180, out=90, looseness=1.25] (31) to (36.center);
		\draw (33) to (38);
		\draw [in=0, out=0, looseness=0.75] (36.center) to (39.center);
		\draw [style=diredge] (37.center) to (33);
		\draw [style=diredge, in=-75, out=180] (39.center) to (34);
		\draw [style=diredge, in=-45, out=90] (34) to (31);
		\draw [style=diredge, in=-153, out=90] (38) to (34);
		\draw [style=diredge, in=-135, out=60, looseness=0.75] (32.center) to (31);
		\draw [style=diredge, in=180, out=90] (48) to (43.center);
		\draw (45) to (44);
		\draw [style=diredge, in=-75, out=180] (40.center) to (48);
		\draw [in=0, out=0, looseness=1.50] (43.center) to (40.center);
		\draw [style=diredge, in=-150, out=75, looseness=1.25] (41.center) to (46);
		\draw [style=diredge, in=-135, out=90] (46) to (48);
		\draw [style=diredge, in=-45, out=90, looseness=1.25] (44) to (46);
		\draw [style=diredge] (49.center) to (45);
		\draw [style=diredge, in=180, out=90] (55) to (52.center);
		\draw [style=diredge, in=-117, out=90] (57) to (55);
		\draw (58) to (57);
		\draw [style=diredge, in=-75, out=180] (54.center) to (55);
		\draw [in=0, out=0, looseness=1.50] (52.center) to (54.center);
		\draw [style=diredge, in=-150, out=60, looseness=1.50] (51.center) to (56);
		\draw [style=diredge, in=-30, out=120, looseness=1.50] (53.center) to (56);
		\draw [style=diredge] (56) to (58);
		\draw [style=diredge, in=-150, out=60, looseness=1.50] (61.center) to (59);
		\draw [style=diredge, in=-30, out=120, looseness=1.50] (60.center) to (59);
		\draw [style=diredge] (59) to (62);
	\end{pgfonlayer}
\end{tikzpicture}}
\endpgfgraphicnamed\qedhere\]
\end{proof}

\begin{definition}\label{def:pos-dim}
  An object $X$ in a dagger compact category is
  \textit{positive-dimensional} if there is a positive definite scalar
  $z$ satisfying $\;\begin{aligned}%
\beginpgfgraphicnamed{pos_dim}
\begin{tikzpicture}[dotpic, font=\scriptsize]
	\begin{pgfonlayer}{nodelayer}
		\node [style=none] (0) at (-3, 2.25) {};
		\node [style=none] (1) at (-3, 0) {};
		\node [style=none] (2) at (0, 1.75) {};
		\node [style=none] (3) at (0.75, 1) {};
		\node [style=none] (4) at (0, 0.25) {};
		\node [style=none] (5) at (-0.75, 1) {};
		\node [style=square box] (6) at (-1.75, 0.5) {$z$};
		\node [style=square box] (7) at (-1.75, 1.75) {$z$};
		\node [style=none] (8) at (2, 1) {$=$};
		\node [style=none] (9) at (3.5, 2.25) {};
		\node [style=none] (10) at (3.5, 0) {};
		\node [style=none] (11) at (0.75, 0.25) {$X$};
		\node [style=none] (12) at (4, 0.25) {$X$};
		\node [style=none] (13) at (-3.5, 0.25) {$X$};
	\end{pgfonlayer}
	\begin{pgfonlayer}{edgelayer}
		\draw [style=diredge] (1.center) to (0.center);
		\draw [in=90, out=0] (2.center) to (3.center);
		\draw [in=0, out=-90] (3.center) to (4.center);
		\draw [in=-90, out=180] (4.center) to (5.center);
		\draw [style=diredge, in=180, out=90] (5.center) to (2.center);
		\draw [style=diredge] (10.center) to (9.center);
	\end{pgfonlayer}
\end{tikzpicture}}
\endpgfgraphicnamed\end{aligned}$.
  A dagger compact closed category is called positive-dimensional if
  all its objects are.
\end{definition}


\begin{proposition}\label{prop:CPMobs}
  For a positive-dimensional dagger compact closed category \V, every
  object of the form $A^* \otimes A$ carries a canonical
  normalisable dagger Frobenius algebra, given as follows.
  \ctikzfig{pants_alg}
\end{proposition}
\begin{proof}
  The Frobenius axioms follow immediately from compact closure. 
  By positive-dimensionality, there exists a positive-definite scalar
  $z$ such that $(z^2 \circ \Tr_A(\id[A])) \otimes \id[A] = \id[A]$. It is then
  possible to show that $\id[A^*\otimes A] \otimes z$ is the
  normaliser. 
\end{proof}

From now on, we will always take \V to be positive-dimensional.

\section{The CP*--construction}\label{sec:CPs}
This section defines the CP*--construction. In fact, defining it is
easy; most work goes into proving that it yields a dagger compact
category. For the definition we need some (graphical) notation,
generalising the left and right regular actions $A \mapsto
\mathop{End}(A)$ given by $x \mapsto x \cdot (-)$ and $x \mapsto 
(-) \cdot x$ for a finite-dimensional algebra $A$. More 
generally, for a monoid $(A,\whitemult)$ in a compact category, define its left
and right action maps $A \to A^* \otimes A$ as follows.
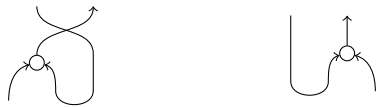
\ctikzfig{action_maps}
Similarly, we can define for any comonoid left and right coaction
maps. By convention, we work primarily with right action and coaction
maps, and express them more succinctly as follows.
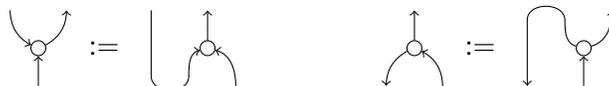
\ctikzfig{action_abbrev}


\begin{definition}
  Let \V be a dagger compact category. Objects of $\CPs[\V]$ are
  normalizable dagger Frobenius algebras in \V. Morphisms from
  $(A,\whitemult)$ to $(B,\graymult)$ in $\CPs[\V]$ are morphisms $f
  \colon A \to B$ in \V such that there exists an object $X$ and $g \colon
  A \rightarrow X \otimes B$ in \V satisfying the following equation.
  \begin{equation}\label{eq:cpstar-condition}
\beginpgfgraphicnamed{cpstar_condition}
\begin{tikzpicture}[dotpic]
	\begin{pgfonlayer}{nodelayer}
		\node [style=none] (0) at (0, -0) {$=$};
		\node [style=square box] (1) at (-2, -0) {$f$};
		\node [style=gray dot] (2) at (-2, 1) {};
		\node [style=none] (3) at (-1.25, 2) {};
		\node [style=none] (4) at (-2.75, 2) {};
		\node [style=white dot] (5) at (-2, -1) {};
		\node [style=none] (6) at (-1.25, -2) {};
		\node [style=none] (7) at (-2.75, -2) {};
		\node [style=square box, minimum width=1 cm, minimum height=0.75 cm] (8) at (2, -0) {$g_*$};
		\node [style=none] (9) at (4.5, -0.75) {};
		\node [style=none] (10) at (2, -0.75) {};
		\node [style=none] (11) at (2, -1.75) {};
		\node [style=none] (12) at (4.5, -1.75) {};
		\node [style=none] (13) at (5, 0.75) {};
		\node [style=none] (14) at (1.5, 0.75) {};
		\node [style=none] (15) at (5, 2) {};
		\node [style=none] (16) at (1.5, 2) {};
		\node [style=square box, minimum width=1 cm, minimum height=0.75 cm] (17) at (4.5, -0) {$g$};
		\node [style=none] (18) at (2.5, 0.75) {};
		\node [style=none] (19) at (4, 0.75) {};
	\end{pgfonlayer}
	\begin{pgfonlayer}{edgelayer}
		\draw [style=diredge] (1) to (2);
		\draw [style=diredge, bend right] (4.center) to (2);
		\draw [style=diredge, bend right] (2) to (3.center);
		\draw [style=diredge, bend right] (5) to (7.center);
		\draw [style=diredge, bend right] (6.center) to (5);
		\draw [style=diredge] (5) to (1);
		\draw [style=diredge] (12.center) to (9.center);
		\draw [style=diredge] (13.center) to (15.center);
		\draw [style=diredge] (10.center) to (11.center);
		\draw [style=diredge] (16.center) to (14.center);
		\draw [style=diredge, in=90, out=90, looseness=1.50] (19.center) to (18.center);
	\end{pgfonlayer}
\end{tikzpicture}}
\endpgfgraphicnamed
  \end{equation}
  Composition and identities are inherited from \V.
\end{definition}

Equation (\ref{eq:cpstar-condition}) is called the
\emph{CP*--condition}. We first establish that $\CPs[\V]$ is a
well-defined category. 

\begin{lemma}\label{lem:frob-actions}
  Any symmetric Frobenius algebra satisfies $\;%
\beginpgfgraphicnamed{sfa_ident}
\begin{tikzpicture}[dotpic]
	\begin{pgfonlayer}{nodelayer}
		\node [style=white dot] (0) at (-1.25, 0.5) {};
		\node [style=white dot] (1) at (-1.25, -0.25) {};
		\node [style=none] (2) at (-1.75, -0.75) {};
		\node [style=none] (3) at (-0.75, -0.75) {};
		\node [style=none] (4) at (-1.75, 1) {};
		\node [style=none] (5) at (-0.75, 1) {};
		\node [style=none] (6) at (2.5, 1) {};
		\node [style=none] (7) at (2.5, -0.75) {};
		\node [style=white dot] (8) at (1.25, -0) {};
		\node [style=none] (9) at (1.25, -0.75) {};
		\node [style=white dot] (10) at (2.5, -0) {};
		\node [style=none] (11) at (1.25, 1) {};
		\node [style=none] (12) at (0, -0) {$=$};
	\end{pgfonlayer}
	\begin{pgfonlayer}{edgelayer}
		\draw [style=diredge, bend right=15] (4.center) to (0);
		\draw [style=diredge, bend right=15] (0) to (5.center);
		\draw [style=diredge, bend right=15] (3.center) to (1);
		\draw [style=diredge, bend right=15] (1) to (2.center);
		\draw [style=diredge] (1) to (0);
		\draw [style=diredge] (11.center) to (8);
		\draw [style=diredge] (8) to (9.center);
		\draw [style=diredge] (7.center) to (10);
		\draw [style=diredge] (10) to (6.center);
		\draw [style=diredge, bend right=60, looseness=1.25] (10) to (8);
	\end{pgfonlayer}
\end{tikzpicture}}
\endpgfgraphicnamed\;$.
\end{lemma}
\begin{proof}
  We can use symmetry to prove the result.
  \[%
\beginpgfgraphicnamed{sfa_ident_pf}
\begin{tikzpicture}[dotpic]
	\begin{pgfonlayer}{nodelayer}
		\node [style=white dot] (0) at (-8.5, 0.5) {};
		\node [style=white dot] (1) at (-8.5, -0.5) {};
		\node [style=none] (2) at (-9, -1.25) {};
		\node [style=none] (3) at (-8, -1.25) {};
		\node [style=none] (4) at (-9, 1.25) {};
		\node [style=none] (5) at (-8, 1.25) {};
		\node [style=none] (6) at (-7.5, -0) {$=$};
		\node [style=none] (7) at (-6.75, -1.25) {};
		\node [style=none] (8) at (-6, -1.25) {};
		\node [style=white dot] (9) at (-6, 0.5) {};
		\node [style=white dot] (10) at (-6, -0.5) {};
		\node [style=none] (11) at (-6, 1.25) {};
		\node [style=none] (12) at (-6.75, 1.25) {};
		\node [style=none] (13) at (-5, -0) {$=$};
		\node [style=white dot] (14) at (-3, -0.75) {};
		\node [style=none] (15) at (-3.75, 1.5) {};
		\node [style=none] (16) at (-4.5, -1.5) {};
		\node [style=none] (17) at (-2, 1.5) {};
		\node [style=none] (18) at (-3, -1.5) {};
		\node [style=white dot] (19) at (-2.5, 0.5) {};
		\node [style=white dot] (20) at (-3, 1.25) {};
		\node [style=none] (21) at (-4, 0.25) {};
		\node [style=none] (22) at (-0.5, 1.5) {};
		\node [style=none] (23) at (-1.25, -1.5) {};
		\node [style=none] (24) at (1.25, 1.5) {};
		\node [style=none] (25) at (-0.75, 0.5) {};
		\node [style=none] (26) at (-1.75, -0) {$=$};
		\node [style=white dot] (27) at (0.5, -0.75) {};
		\node [style=white dot] (28) at (0, -0) {};
		\node [style=none] (29) at (0.5, -1.5) {};
		\node [style=white dot] (30) at (0.25, 1.25) {};
		\node [style=white dot] (31) at (3.25, -0) {};
		\node [style=none] (32) at (1.75, -0) {$=$};
		\node [style=none] (33) at (2.75, 0.5) {};
		\node [style=white dot] (34) at (3.75, -0.75) {};
		\node [style=none] (35) at (4.5, 1.5) {};
		\node [style=none] (36) at (5, 1.5) {};
		\node [style=white dot] (37) at (3.75, 1.25) {};
		\node [style=none] (38) at (2.25, -1.5) {};
		\node [style=none] (39) at (3.75, -1.5) {};
		\node [style=none] (40) at (7.5, 1.5) {};
		\node [style=white dot] (41) at (7.25, -0.75) {};
		\node [style=none] (42) at (5.75, -1.5) {};
		\node [style=none] (43) at (8.5, 1.5) {};
		\node [style=white dot] (44) at (6.75, 0.5) {};
		\node [style=none] (45) at (5.25, -0) {$=$};
		\node [style=none] (46) at (6.25, 1) {};
		\node [style=none] (47) at (7.25, -1.5) {};
		\node [style=none] (48) at (8.75, -0) {$=$};
		\node [style=none] (49) at (10.75, -1.25) {};
		\node [style=none] (50) at (9.5, 1.25) {};
		\node [style=none] (51) at (10.75, 1.25) {};
		\node [style=white dot] (52) at (9.5, -0) {};
		\node [style=none] (53) at (9.5, -1.25) {};
		\node [style=white dot] (54) at (10.75, -0) {};
	\end{pgfonlayer}
	\begin{pgfonlayer}{edgelayer}
		\draw [style=diredge, bend right=15] (4.center) to (0);
		\draw [style=diredge, bend right=15] (0) to (5.center);
		\draw [style=diredge, bend right=15] (3.center) to (1);
		\draw [style=diredge, bend right=15] (1) to (2.center);
		\draw [style=diredge] (1) to (0);
		\draw [style=diredge, in=-135, out=-90, looseness=1.50] (12.center) to (9);
		\draw [style=diredge] (9) to (11.center);
		\draw [style=diredge] (8.center) to (10);
		\draw [style=diredge, in=90, out=135, looseness=1.50] (10) to (7.center);
		\draw [style=diredge] (10) to (9);
		\draw [style=diredge, bend right=15] (19) to (17.center);
		\draw [style=diredge] (18.center) to (14);
		\draw [style=diredge] (14) to (19);
		\draw [style=diredge, in=-120, out=-90, looseness=2.50] (15.center) to (20);
		\draw [in=0, out=135, looseness=0.75] (14) to (21.center);
		\draw [style=diredge, in=90, out=180, looseness=0.75] (21.center) to (16.center);
		\draw [style=diredge] (19) to (20);
		\draw [style=diredge, in=-120, out=-90, looseness=2.50] (22.center) to (30);
		\draw [in=0, out=135, looseness=0.75] (28) to (25.center);
		\draw [style=diredge, in=90, out=180, looseness=0.75] (25.center) to (23.center);
		\draw [style=diredge, in=-60, out=45] (28) to (30);
		\draw [style=diredge] (29.center) to (27);
		\draw [style=diredge] (27) to (28);
		\draw [style=diredge, bend right=15] (27) to (24.center);
		\draw [style=diredge, in=-60, out=-90, looseness=2.50] (35.center) to (37);
		\draw [in=0, out=135, looseness=0.75] (31) to (33.center);
		\draw [style=diredge, in=90, out=180, looseness=0.75] (33.center) to (38.center);
		\draw [style=diredge, in=-135, out=45] (31) to (37);
		\draw [style=diredge] (39.center) to (34);
		\draw [style=diredge] (34) to (31);
		\draw [style=diredge, bend right=15] (34) to (36.center);
		\draw [in=0, out=135, looseness=0.75] (44) to (46.center);
		\draw [style=diredge, in=90, out=180, looseness=0.50] (46.center) to (42.center);
		\draw [style=diredge] (47.center) to (41);
		\draw [style=diredge, in=-135, out=142, looseness=1.25] (41) to (44);
		\draw [style=diredge, bend right=15] (41) to (43.center);
		\draw [style=diredge, in=-45, out=-90, looseness=1.50] (40.center) to (44);
		\draw [style=diredge] (50.center) to (52);
		\draw [style=diredge] (52) to (53.center);
		\draw [style=diredge] (49.center) to (54);
		\draw [style=diredge] (54) to (51.center);
		\draw [style=diredge, bend right=60, looseness=1.25] (54) to (52);
	\end{pgfonlayer}
\end{tikzpicture}}
\endpgfgraphicnamed\qedhere\]
\end{proof}


\begin{lemma}\label{lem:norm-alt}
  Any normalisable dagger Frobenius algebra satisfies $\;%
\beginpgfgraphicnamed{norm_alt}
\begin{tikzpicture}[dotpic]
	\begin{pgfonlayer}{nodelayer}
		\node [style=none] (0) at (1, -1.25) {};
		\node [style=none] (1) at (1, 1.25) {};
		\node [style=none] (2) at (-1.5, 1.5) {};
		\node [style=none] (3) at (-1.5, -1.5) {};
		\node [style=none] (4) at (0, -0) {$=$};
		\node [style=white dot] (5) at (-1.5, 0.75) {};
		\node [style=white dot] (6) at (-1.5, -0.75) {};
		\node [style=white norm] (7) at (-2, -0) {};
		\node [style=white norm] (8) at (-1, -0) {};
	\end{pgfonlayer}
	\begin{pgfonlayer}{edgelayer}
		\draw [style=diredge] (0.center) to (1.center);
		\draw [style=diredge] (5) to (2.center);
		\draw [style=diredge] (3.center) to (6);
		\draw [style=diredge, in=-90, out=30] (6) to (8);
		\draw [style=diredge, in=-30, out=90] (8) to (5);
		\draw [style=diredge, in=90, out=-150] (5) to (7);
		\draw [style=diredge, in=150, out=-90] (7) to (6);
	\end{pgfonlayer}
\end{tikzpicture}}
\endpgfgraphicnamed\;$.
\end{lemma}
\begin{proof}
  Follows from Frobenius equations and a trace identity $(*)$ for
  symmetric Frobenius algebras.
  \[%
\beginpgfgraphicnamed{norm_alt_pf}
\begin{tikzpicture}[dotpic]
	\begin{pgfonlayer}{nodelayer}
		\node [style=none] (0) at (-1.5, 1.75) {};
		\node [style=none] (1) at (-1.5, -1.75) {};
		\node [style=none] (2) at (0, -0) {$=$};
		\node [style=white dot] (3) at (-1.5, 1) {};
		\node [style=white dot] (4) at (-1.5, -1) {};
		\node [style=white norm] (5) at (-2, -0) {};
		\node [style=white norm] (6) at (-1, -0) {};
		\node [style=white dot] (7) at (1.25, 1.25) {};
		\node [style=none] (8) at (1.25, -2) {};
		\node [style=white dot] (9) at (1.25, -0) {};
		\node [style=none] (10) at (1.25, 2) {};
		\node [style=white norm] (11) at (1.25, -1.25) {};
		\node [style=white norm] (12) at (1.25, -0.75) {};
		\node [style=none] (13) at (4.5, -2) {};
		\node [style=white dot] (14) at (4.25, 1) {};
		\node [style=none] (15) at (4.75, 2) {};
		\node [style=white norm] (16) at (4.5, -0.75) {};
		\node [style=white norm] (17) at (4.5, -1.25) {};
		\node [style=white dot] (18) at (4.25, -0) {};
		\node [style=none] (19) at (2.5, -0) {$=$};
		\node [style=none] (20) at (5.75, -0) {$=$};
		\node [style=white dot] (21) at (8, -0) {};
		\node [style=white norm] (22) at (8, -0.75) {};
		\node [style=none] (23) at (8.75, 2) {};
		\node [style=none] (24) at (8, -2) {};
		\node [style=white norm] (25) at (8, -1.25) {};
		\node [style=none] (26) at (6.75, 0.5) {};
		\node [style=white dot] (27) at (7.25, 1) {};
		\node [style=none] (28) at (6.75, 1.75) {};
		\node [style=none] (29) at (9.5, -0) {$=$};
		\node [style=none] (30) at (10.5, 1.75) {};
		\node [style=white norm] (31) at (11.25, -0.25) {};
		\node [style=none] (32) at (10.5, 0.5) {};
		\node [style=white norm] (33) at (11.25, 0.25) {};
		\node [style=none] (34) at (12.5, 2) {};
		\node [style=white dot] (35) at (11.75, -1) {};
		\node [style=none] (36) at (11.75, -2) {};
		\node [style=white dot] (37) at (11, 1) {};
		\node [style=none] (38) at (16.75, -0) {$=$};
		\node [style=white dot] (39) at (17.5, 0.5) {};
		\node [style=white dot] (40) at (18, -0.75) {};
		\node [style=none] (41) at (18.75, 1.75) {};
		\node [style=none] (42) at (18, -1.75) {};
		\node [style=none] (43) at (19.25, -0) {$=$};
		\node [style=none] (44) at (20.25, -1.25) {};
		\node [style=none] (45) at (20.25, 1.25) {};
		\node [style=none] (46) at (3.75, 1.5) {};
		\node [style=none] (47) at (3.75, -0.5) {};
		\node [style=none] (48) at (13.25, -0) {$=$};
		\node [style=white norm] (49) at (14.25, 0.5) {};
		\node [style=none] (50) at (16.25, 2) {};
		\node [style=white dot] (51) at (15.25, -1) {};
		\node [style=none] (52) at (15.25, -2) {};
		\node [style=none] (53) at (15, 2) {};
		\node [style=white dot] (54) at (14.5, 1.25) {};
		\node [style=none] (55) at (15, 0.75) {};
		\node [style=white norm] (56) at (14.25, -0) {};
		\node [style=none] (57) at (13.25, 0.75) {\small $(*)$};
	\end{pgfonlayer}
	\begin{pgfonlayer}{edgelayer}
		\draw [style=diredge] (3) to (0.center);
		\draw [style=diredge] (1.center) to (4);
		\draw [style=diredge, in=-90, out=30] (4) to (6);
		\draw [style=diredge, in=-30, out=90] (6) to (3);
		\draw [style=diredge, in=90, out=-150] (3) to (5);
		\draw [style=diredge, in=150, out=-90] (5) to (4);
		\draw [style=diredge] (7) to (10.center);
		\draw [style=diredge] (8.center) to (11);
		\draw (11) to (12);
		\draw [style=diredge] (12) to (9);
		\draw [style=diredge, bend right=45, looseness=1.25] (9) to (7);
		\draw [style=diredge, bend right=45, looseness=1.25] (7) to (9);
		\draw [style=diredge] (14) to (15.center);
		\draw [style=diredge] (13.center) to (17);
		\draw (17) to (16);
		\draw [style=diredge, in=-72, out=90] (16) to (18);
		\draw [style=diredge] (18) to (14);
		\draw [style=diredge, bend right=15] (21) to (23.center);
		\draw [style=diredge] (24.center) to (25);
		\draw (25) to (22);
		\draw [style=diredge] (22) to (21);
		\draw [in=0, out=90] (27) to (28.center);
		\draw [style=diredge, in=-105, out=0] (26.center) to (27);
		\draw [in=180, out=180, looseness=1.50] (28.center) to (26.center);
		\draw [style=diredge, in=-45, out=135] (21) to (27);
		\draw [style=diredge, in=-90, out=45, looseness=0.75] (35) to (34.center);
		\draw (31) to (33);
		\draw [in=0, out=90] (37) to (30.center);
		\draw [style=diredge, in=-105, out=0] (32.center) to (37);
		\draw [in=180, out=180, looseness=1.50] (30.center) to (32.center);
		\draw [style=diredge, in=-90, out=135] (35) to (31);
		\draw [style=diredge, in=-60, out=90] (33) to (37);
		\draw [style=diredge] (36.center) to (35);
		\draw [style=diredge, in=-90, out=45, looseness=0.75] (40) to (41.center);
		\draw [style=diredge] (42.center) to (40);
		\draw [style=diredge, in=-90, out=142] (40) to (39);
		\draw [style=diredge] (44.center) to (45.center);
		\draw [in=0, out=120] (14) to (46.center);
		\draw [in=180, out=180] (46.center) to (47.center);
		\draw [style=diredge, in=-120, out=0] (47.center) to (18);
		\draw [style=diredge, in=-90, out=45, looseness=0.75] (51) to (50.center);
		\draw (56) to (49);
		\draw [in=180, out=90] (54) to (53.center);
		\draw [style=diredge, in=-75, out=180] (55.center) to (54);
		\draw [in=0, out=0, looseness=1.50] (53.center) to (55.center);
		\draw [style=diredge, in=-90, out=135] (51) to (56);
		\draw [style=diredge, in=-120, out=90] (49) to (54);
		\draw [style=diredge] (52.center) to (51);
	\end{pgfonlayer}
\end{tikzpicture}}
\endpgfgraphicnamed\qedhere\]
\end{proof}

These lemmas can express the CP*--condition in the sometimes more
convenient ``convolution form''.

\begin{proposition}
  Let $\V$ be a dagger compact category, $(A, \whitemult)$ and $(B,
  \graymult)$ be normalisable dagger Frobenius algebras, and $f \colon
  A \to B$ a morphism.
  \ctikzfig{cpstar_conv_form}
\end{proposition}
\begin{proof}
  For ($\Rightarrow$), apply $(\graycoaction \circ - \circ
  \whiteaction)$ to both sides and use Lemma~\ref{lem:norm-alt} and
  properties of normalisers. For ($\Leftarrow$), apply $(\grayaction
  \circ - \circ \whitecoaction)$ to both sides and apply
  Lemma~\ref{lem:frob-actions}. 
\end{proof}

\begin{theorem}\label{thm:cpstar-smc}
  If \V is a dagger compact category \V, so is $\CPs[\V]$.
\end{theorem}
\begin{proof}
  Identity maps $\id[A] \colon (A, \whitemult) \to (A,
  \whitemult)$ satisfy the $\CPs$-condition: letting $g=\whitecomult$
  in~\eqref{eq:cpstar-condition} does the job by
  Lemma~\ref{lem:frob-actions}, whose left-hand side is$\whiteaction
  \circ \id[A] \circ \whitecoaction$.
  
  Next, suppose $f \colon (A, \whitemult) \to (B, \graymult)$ and
  $g \colon (B, \graymult) \to (C, \blackmult)$ satisfy the
  CP*--condition. It then follows from Lemma~\ref{lem:norm-alt} that
  their composition does, too.
  \ctikzfig{compose_cpstar}
  
  For the monoidal structure, take $(A,\whitemult) \otimes (B,
  \graymult) := (A \otimes B, \prodmult{white dot}{gray dot})$. For
  maps $f \colon (A, \whitemult) \to (C, \dotmult{alt white dot})$ and
  $g \colon (B, \graymult) \to (D, \dotmult{alt gray dot})$
  satisfying~\eqref{eq:cpstar-condition}, also  $f \otimes g \colon (A
  \otimes B, \prodmult{white dot}{gray dot}) \to (C \otimes D,
  \prodmult{alt white dot}{alt gray dot})$ satisfies the
  CP*--condition. This can be seen by applying the coaction of
  $\prodmult{white dot}{gray dot}$ and the action of $\prodmult{alt
  white dot}{alt gray dot}$, then decomposing $f$ into $h_*,h$ and $g$
  into $i_*,i$, as follows.
  \ctikzfig{cpstar_monoidal}
  
  As for the tensor unit, note that $I := (I, \rho_I)$ is a
  normalisable dagger Frobenius algebra by monoidal coherence in \V.
  Using this definition of $\otimes$ and $I$, the monoidal structure
  maps $\alpha$, $\lambda$, and $\rho$ from \V trivially satisfy the
  CP*--condition. Thus $\CPs[\V]$ is a monoidal category. $\CPs[\V]$
  inherits the dagger from \V. Symmetry and dual maps in \V lift to the
  following morphisms in $\CPs[\V]$.
  \begin{center}
    $\sigma_{A,B} \colon
       (A \otimes B, \prodmult{white dot}{gray dot}) \to
       (B \otimes A, \prodmult{gray dot}{white dot})$
    \qquad\qquad
    $ e_{A^*} \colon (A, \whitemult) \otimes (A^*, \whitedualmult) \to I $
  \end{center}
  The Frobenius identities and Lemma~\ref{lem:frob-actions} establish
  the CP*--condition for these maps.
\end{proof}

We refer to a morphism $ I \rightarrow (A, \whitemult)$ of $\CPs[\V]$ as a
\textit{positive element} of $(A, \whitemult)$. Another way to express
the CP*--condition for a \V-morphism is to say that
it preserves the property of being a positive element when applied to a
some subsystem, as in the following theorem. This will be useful to
connect to the traditional notion of complete positivity of linear
maps between C*-algebras.  

\begin{theorem}\label{thm:pos-elems}
  Let $(A,\whitemult)$ and $(B, \graymult)$ be normalisable dagger
  Frobenius algebras and $f \colon A \to B$ a morphism in a dagger
  compact category $\V$. The following are equivalent:
  \begin{enumerate}[(a)]
  \item $f$ satisfies the CP*--condition;
  \item postcomposing with $f \otimes \id[C]$ sends positive elements
    of $(A,\whitemult) \otimes (C, \blackmult)$ to positive elements
    of $(B, \graymult) \otimes (C, \blackmult)$ for all dagger
    normalisable Frobenius algebras $(C, \blackmult)$;
  \item postcomposing with $f \otimes \id[X^* \otimes X]$ sends positive elements
    of $(A,\whitemult) \otimes (X^* \otimes X,\pantsalg)$ to positive
    elements of $(B, \graymult) \otimes (X^* \otimes X,\pantsalg)$ for 
    all objects $X$ in $\V$.
  \end{enumerate}
\end{theorem}
\begin{proof}
  For (a) $\Rightarrow$ (b): if $\rho$ is a positive element of
  $(A,\whitemult) \otimes (C,\blackmult)$ and $f$ satisfies the
  CP*--condition, then so does $(f \otimes \id[C]) \circ \rho$, by
  Theorem~\ref{thm:cpstar-smc}. The implication (b) $\Rightarrow$ (c)
  is trivial. Finally, for (c) $\Rightarrow$ (a), take $(C, \blackmult) = (A^*,
  \whitedualmult)$. The action map $\blackaction \colon (C,\blackmult)
  \to (C^* \otimes C, \pantsalg)$ is a morphism in $\CPs[\V]$. As a
  consequence of this fact and Theorem~\ref{thm:cpstar-smc}, the
  following is a positive element of $(A,\whitemult) \otimes (X^*
  \otimes X, \pantsalg)$.
  \ctikzfig{cpstar_alt_matrix_point}
  So, by assumption, $(f \otimes \id[A^*]) \circ \rho$  is also a positive
  element. Applying white caps to both sides finishes the proof.
  \[ %
\beginpgfgraphicnamed{cpstar_alt_pf}
\begin{tikzpicture}[dotpic]
	\begin{pgfonlayer}{nodelayer}
		\node [style=none] (0) at (-7.5, -0) {$=$};
		\node [style=square box] (1) at (-10.25, -0.5) {$f$};
		\node [style=gray dot] (2) at (-10.25, 1.5) {};
		\node [style=none] (3) at (-9.5, 2.5) {};
		\node [style=none] (4) at (-11, 2.5) {};
		\node [style=white dot] (5) at (-9, -0) {};
		\node [style=none] (6) at (-11.75, 1.75) {};
		\node [style=none] (7) at (-8.75, 2.5) {};
		\node [style=square box, minimum width=1 cm, minimum height=0.75 cm] (8) at (-5.5, -0.25) {$g_*$};
		\node [style=none] (9) at (-2.75, 0.5) {};
		\node [style=none] (10) at (-6.25, 0.5) {};
		\node [style=none] (11) at (-2.75, 1.75) {};
		\node [style=none] (12) at (-6.25, 1.75) {};
		\node [style=square box, minimum width=1 cm, minimum height=0.75 cm] (13) at (-3, -0.25) {$g$};
		\node [style=none] (14) at (-5, 0.5) {};
		\node [style=none] (15) at (-3.5, 0.5) {};
		\node [style=none] (16) at (-10.25, -1.5) {};
		\node [style=none] (17) at (-9, -1.5) {};
		\node [style=none] (18) at (-11.75, 2.5) {};
		\node [style=none] (19) at (-2.25, 0.5) {};
		\node [style=none] (20) at (-2.25, 1.75) {};
		\node [style=none] (21) at (-5.75, 1.75) {};
		\node [style=none] (22) at (-5.75, 0.5) {};
		\node [style=none] (23) at (0, -0) {$\Rightarrow$};
		\node [style=none] (24) at (4.25, 2) {};
		\node [style=none] (25) at (11, 0.5) {};
		\node [style=none] (26) at (4.25, -2) {};
		\node [style=none] (27) at (2.75, -2) {};
		\node [style=square box, minimum width=1 cm, minimum height=0.75 cm] (28) at (8, -0.25) {$g_*$};
		\node [style=square box, minimum width=1 cm, minimum height=0.75 cm] (29) at (10.5, -0.25) {$g$};
		\node [style=white dot] (30) at (3.5, -1) {};
		\node [style=white dot] (31) at (7, 1) {};
		\node [style=white dot] (32) at (11.5, 1) {};
		\node [style=none] (33) at (10, 0.5) {};
		\node [style=none] (34) at (2.75, 2) {};
		\node [style=none] (35) at (8.5, 0.5) {};
		\node [style=none] (36) at (5.5, -0) {$=$};
		\node [style=gray dot] (37) at (3.5, 1) {};
		\node [style=none] (38) at (7.5, 0.5) {};
		\node [style=none] (39) at (10.5, 2.25) {};
		\node [style=square box] (40) at (3.5, -0) {$f$};
		\node [style=none] (41) at (8, 2.25) {};
		\node [style=none] (42) at (10.5, 0.5) {};
		\node [style=none] (43) at (8, 0.5) {};
		\node [style=none] (44) at (6.5, 0.5) {};
		\node [style=none] (45) at (12, 0.5) {};
		\node [style=none] (46) at (6.5, -2) {};
		\node [style=none] (47) at (12, -2) {};
	\end{pgfonlayer}
	\begin{pgfonlayer}{edgelayer}
		\draw [style=diredge] (1) to (2);
		\draw [style=diredge, in=143, out=-90] (4.center) to (2);
		\draw [style=diredge, in=-90, out=37] (2) to (3.center);
		\draw [style=diredge, in=-90, out=150] (5) to (7.center);
		\draw [style=diredge, in=60, out=-90] (6.center) to (5);
		\draw [style=diredge] (9.center) to (11.center);
		\draw [style=diredge] (12.center) to (10.center);
		\draw [style=diredge, in=90, out=90, looseness=1.50] (15.center) to (14.center);
		\draw [style=diredge] (16.center) to (1);
		\draw [in=-90, out=-90, looseness=1.25] (17.center) to (16.center);
		\draw (5) to (17.center);
		\draw (18.center) to (6.center);
		\draw [style=diredge] (19.center) to (20.center);
		\draw [style=diredge] (21.center) to (22.center);
		\draw [style=diredge] (40) to (37);
		\draw [style=diredge, bend right] (34.center) to (37);
		\draw [style=diredge, bend right] (37) to (24.center);
		\draw [style=diredge, bend right] (30) to (27.center);
		\draw [style=diredge, bend right] (26.center) to (30);
		\draw [style=diredge] (30) to (40);
		\draw [style=diredge] (42.center) to (39.center);
		\draw [style=diredge, in=-165, out=90] (25.center) to (32);
		\draw [style=diredge] (41.center) to (43.center);
		\draw [style=diredge, in=90, out=-15] (31) to (38.center);
		\draw [style=diredge, in=90, out=90, looseness=1.50] (33.center) to (35.center);
		\draw [in=90, out=-165] (31) to (44.center);
		\draw [style=diredge, in=-15, out=90] (45.center) to (32);
		\draw [style=diredge] (44.center) to (46.center);
		\draw (47.center) to (45.center);
	\end{pgfonlayer}
\end{tikzpicture}}
\endpgfgraphicnamed \]
  See also~\cite[Proposition~3.4]{journalversion}.
\end{proof}





\section{Embedding Selinger's CPM--construction}\label{sec:embedding}

This section will concentrate on the ``purely quantum'' objects in
$\CPs[\V]$, by proving that the latter embeds $\CPM[\V]$ in a full,
faithful, and strongly dagger symmetric monoidal way. First, we recall
Selinger's CPM--construction~\cite{selinger:completelypositive}.

\begin{definition}
  For a dagger compact category \V, form the dagger compact category
  $\CPM[\V]$ as follows. Objects are the same as those in \V, and
  morphisms $f \in \CPM[\V](A,B)$ are \V-morphisms $f \colon A^*
  \otimes A \rightarrow B^* \otimes B$ for which there is $g \colon A
  \rightarrow X \otimes B$ satisfying the following condition.
  \ctikzfig{cp_condition}
  Composition, identity maps, and $\otimes$ on objects are defined as
  in \V. On morphisms, $\otimes$ is defined as: 
  \ctikzfig{cp_monoidal}
\end{definition}

A strongly dagger symmetric monoidal functor is a functor $F$ along
with a unitary natural isomorphism $\varphi_{A,B} \colon F(A \otimes B)
\to F(A) \otimes F(B)$ satisfying several coherence properties
that we have no space to go in to.
Our next theorem shows that $\CPM[\V]$ is equivalent to the full subcategory
of $\CPs[\V]$ consisting of objects of the form $(A^* \otimes A,
\pantsalg)$. Its proof uses $*$-homomorphisms, which we first define.

\begin{definition}\label{def:starhomo}
  If $(A, \whitemult)$ and $(B, \graymult)$ are normalisable dagger
  Frobenius algebras in a dagger compact category $\V$, a morphism $f \colon A \to B$ is
  called a \emph{$*$-homomorphism} when it satisfies the following equations.
  \ctikzfig{starhomomorphism}
\end{definition}

\begin{lemma}\label{lem:starhomoiscp}
  Let $(A,\whitemult)$ and $(B,\graymult)$ be normalisable dagger
  Frobenius algebras in a dagger compact category $\V$. If $f \colon A
  \to B$ is a $*$-homomorphism, then it is a well-defined morphism in
  $\CPs[\V]$. 
\end{lemma}

\begin{proof}
  Using the definition: $\grayaction \circ f \circ \whitecoaction = 
  (\grayaction \circ \graycoaction) \circ (f^* \otimes f)$. Applying
  Lemma~\ref{lem:frob-actions} completes the proof.
\end{proof}


\begin{theorem}\label{thm:embedding}
  Let \V be a positive-dimensional dagger compact category. 
  Define $L \colon \CPM[\V] \to \CPs[\V]$ by setting $L(A) := (A^*
  \otimes A, \pantsalg)$ on objects and $L(f) = f$ on morphisms. 
  Then $L$ is a well-defined functor that is full,
  faithful, and strongly dagger symmetric monoidal. 
\end{theorem}
\begin{proof}
  For well-definedness, we show 
  that a \V-morphism $f : A^* \otimes A \to B^* \otimes B$ is a
  $\CPM[\V]$-morphism from $A$ to $B$ if and only if it is a
  $\CPs[\V]$-morphism from $(A^* \otimes A, \pantsalg)$ to $(B^*
  \otimes B, \pantsalg)$. First, assume $f \in \CPM[\V]$ and
  compose with the action and coaction of the respective algebras to
  see that $f$ satisfies the CP*--condition, as follows.
  \ctikzfig{cp_implies_cpstar}
  Conversely, if $f$ is in $\CPs[\V]$, then it is also in $\CPM[\V]$,
  as follows.
  \ctikzfig{cpstar_implies_cp}
  
  Composition is the same in $\CPM[\V]$ and $\CPs[\V]$, so
  $L$ is a functor that is furthermore full and faithful. Similarly,
  daggers are the same in $\CPM[\V]$ and $\CPs[\V]$, so $L$ trivially
  preserves daggers. It now suffices to show that $L$ is strongly
  monoidal. Define the isomorphism $\varphi_{A,B} \colon L(A
  \otimes B) \to L(A) \otimes L(B)$ as the reshuffling map $(321)(4)
  \colon B^* \otimes A^* \otimes A \otimes B \to A^* \otimes A \otimes
  B^* \otimes B$.
  One can verify that this map is a $*$-homomorphism from $L(A\otimes
  B)$ to $L(A) \otimes L(B)$, so it must satisfy the CP*--condition.
  This map is also unitary, and it is a routine calculation to show
  that it satisfies the coherence equations for a strong symmetric
  monoidal functor.
\end{proof}

\section{Generalised stochastic maps and measurement of quantum states}\label{sec:stoch}

Whereas the previous section focused on ``purely quantum'' objects in
$\CPs[\V]$, this section looks at the ``purely classical'' ones. We
will define a ``purely probabilistic'' category $\Stoch[\V]$, that by
construction embeds into $\CPs[\V]$. Thus objects in $\CPs[\V]$ can
be interpreted as being ``combined classical and quantum''. 
The category $\Stoch[\V]$ was first considered
in~\cite{coeckepaquettepavlovic:structuralism}. It was defined in a slightly 
different form there, but one can prove that this coincides with the following
definition.

\begin{definition}  
  For a dagger compact category $\V$, define $\Stoch[\V]$
  to be the full subcategory of $\CPs[\V]$ consisting of all
  commutative normalisable dagger Frobenius algebras.
\end{definition}

The next proposition justifies why this category is that of
``classical channels''. We call a morphism $f \colon (A,\whitemult)
\to (B,\graymult)$ in $\CPs[\V]$ \emph{normalised} if it preserves
counits: $\graycounit \circ f = \whitecounit$.  
Because commutative finite-dimensional C*-algebras correspond to
finite-dimensional Hilbert spaces with a choice of orthonormal basis,
we can think of the latter as objects of $\Stoch[\FHilb]$~\cite{coeckepavlovicvicary:bases,abramskyheunen:hstar,heunen:complementarity}.
Recall that a stochastic map between finite-dimensional Hilbert spaces
is a matrix with positive real entries whose every column sums to one.

\begin{proposition}
  Normalised morphisms in $\Stoch[\FHilb]$ correspond to stochastic
  maps between finite-dimensional Hilbert spaces.
\end{proposition}
\begin{proof}
  See~\cite[3.2.3 and 2.1.3]{keyl:quantuminformation}.
\end{proof}

For any Frobenius algebra $(A,\blackmult)$ we can consider its \emph{copyable
points}: the morphisms $p \colon I \to A$ that are ``copied'' by the
comultiplication, in the sense that $\blackcomult \circ p = p \otimes
p$. This is especially interesting for commutative normalisable dagger Frobenius
algebras, because in $\FHilb$, copyable points form an orthonormal basis for
$A$. Writing vectors in the basis of classical points, one can show
that the normalised positive elements are precisely those 
vectors with positive coefficients that sum to $1$. Thus, normalised
positive elements of a commutative normalisable dagger Frobenius
algebra can be regarded as probability distributions over its copyable points.

So far we have mostly looked at classical and quantum systems in
isolation. But as they live together in a category $\CPs[\V]$, we can also
consider maps between them.
Consider a normalised morphism $P \colon L(H) \to
(A, \blackmult)$ from a quantum to a classical system. Then $P^\dagger
\circ x_i$ is a positive element of $\blackmult$ for each copyable
point $x_i$. Furthermore, any commutative normalisable dagger
Frobenius algebra in $\FHilb$ satisfies $\blackunit = 
\sum_i x_i$. Since $P$ is normalised, $\sum P^\dagger \circ x_i = P^\dagger \circ \sum x_i = P^\dagger
\circ \blackunit = e_H \colon I \to H^* \otimes H$ is the cup from the compact
structure on $H$. Positive elements in $H^* \otimes H$ represent
positive operators from $H$ to itself, and $e_H$ represents the
identity operator. Thus the set $\{ P^\dagger \circ x_i \}$
corresponds to a positive operator-valued measure (POVM). Furthermore,
for any quantum state $\rho$ (i.e.\ normalised positive element in
$L(H)$), it is straightforward to show that $P \circ \rho$ yields the
probability distribution whose $i$th element is the probability of
getting outcome $x_i$, computed via the Born rule.

The dual notion of a morphism $E \colon (A, \blackmult) \to L(H)$ from
a classical system to a quantum system can be thought of as a
preparation. Or, more precisely, as a map carrying a classical
probability distribution over some fixed set of states to a single
mixed state. Choosing a particular decomposition $E \circ \rho$ for a
quantum state $\rho$ gives us a way to represent quantum ensembles.
See also~\cite[3.2.4]{keylwerner:lectures}.

\section{Hilbert spaces}\label{sec:Hilb}

It is high time we looked at some examples. This section proves that
$\CPs[\FHilb]$ is the category of all finite-dimensional C*-algebras
and completely positive maps. The proof also illuminates
Theorem~\ref{thm:embedding}. Namely, $\CPM[\FHilb]$ has
finite-dimensional C*-\emph{factors} for objects. Recall that a
C*-algebra is a factor when its centre is
1-dimensional. Finite-dimensional C*-algebras in fact 
enjoy an easy structure theorem: they are finite direct sums of 
factors, and factors are precisely matrix
algebras~\cite[Theorem~III.1.1]{davidson:cstar}. The following lemma
recalls the structure of these factors, and the subsequent proposition
determines the objects of $\CPs[\FHilb]$.  

\begin{lemma}\label{lem:Mn}
  If $H$ is an $n$-dimensional Hilbert space, then there is an
  isomorphism of algebras between $L(H)$ and
  $\mathbb{M}_n(\mathbb{C})$. Therefore, $H^* \otimes H$ carries
  C*-algebra structure; the involution is as in
  Theorem~\ref{thm:dfa-cstar}. 
\end{lemma}
\begin{proof}
  First of all, $\mathbb{M}_n(\mathbb{C})$ is a Hilbert space under the
  Hilbert--Schmidt inner product $\inprod{a}{b} = \Tr(a^\dag b)$. It
  has a canonical orthonormal basis $\{e_{ij} \mid i,j=1,\ldots,n\}$,
  where $e_{ij}$ is the matrix all of whose entries are 0 except
  the $(i,j)$-entry, which is 1.
  Pick an orthonormal basis $\{\ket{1}, \ldots, \ket{n}\}$ for $H$, so that
  $\{\bra{i} \otimes \ket{j} \mid i,j=1,\ldots,n\}$ forms an
  orthonormal basis for $H^* \otimes H$. 
  Then the assignment $\bra{i} \otimes \ket{j} \mapsto e_{ij}$
  implements a unitary isomorphism between $H^* \otimes H$ and
  $\mathbb{M}_n(\mathbb{C})$. Direct computation shows that matrix
  multiplication translates across this isomorphism to $\pantsalg$ on
  $H^* \otimes H$. Similarly, taking the conjugate transpose of a
  matrix corresponds to the involution on $H^* \otimes H$ given in
  Theorem~\ref{thm:dfa-cstar}. 
\end{proof}

\begin{proposition}\label{prop:dNFAsinFHilb}
  All dagger Frobenius algebras in \FHilb are normalisable.
\end{proposition}
\begin{proof}
  Let $(A, \whitemult, \whiteunit)$ be a dagger Frobenius
  algebra in $\FHilb$. By Theorem~\ref{thm:dfa-cstar}, it must be
  isomorphic to a C*-algebra of the form $\bigoplus_k
  \mathbb{M}_{n_k}(\mathbb{C})$, giving a unitary isomorphism of the associated
  dagger Frobenius algebras. Let $\{ e_{ij}^{(k)} : 0 \leq i,j < n_k
  \}$ form an orthonormal basis for $A$. We can define $\whitemult$ in 
  terms of this basis as
  $ \whitemult(e_{ij}^{(k)} \otimes e_{i'j'}^{(k')}) =
  \delta_k^{k'} \delta_j^{i'} e_{ij'}^{(k)} $.
  From this, we can compute $\Tr_A(\whitemult)$ directly.
  \[
    \Tr_A(\whitemult)(e_{ij}^{(k)}) = 
    \sum_{i'j'k'}
    \left( e_{i'j'}^{(k')} \right)^\dag
    \whitemult \left(
      e_{ij}^{(k)} \otimes
      e_{i'j'}^{(k')}
    \right)
    =
    \sum_{i'j'k'}
    \left( e_{i'j'}^{(k')} \right)^\dag
    \delta_k^{k'} \delta_j^{i'} e_{ij'}^{(k)} 
    =
    \sum_{j'}
    \left( e_{jj'}^{(k)} \right)^\dag e_{ij'}^{(k)}
    =
    \sum_{j'} \delta_i^j = n_k \delta_i^j.
  \]
  Note that $\whitecounit(e_{ij}^{(k)}) = \delta_i^j$. We can now
  define the normaliser $\whitenorm$ as $e_{ij}^{(k)} \mapsto
  \frac{1}{\sqrt{n_k}} e_{ij}^{(k)}$: this map is invertible, satisfies
  $\Tr_A(\whitemult)\circ (\whitenorm)^2 = \whitecounit$, and
  acts by a constant scalar on each summand of $A$ and so is central. 
\end{proof}

Combining the previous proposition with Theorem~\ref{thm:dfa-cstar},
we see that the objects of $\CPs[\FHilb]$ are (in 1-to-1
correspondence with) finite-dimensional C*-algebras.
We now turn to the morphisms of $\CPs[\FHilb]$. First, let us review
what (concrete) completely positive maps between C*-algebras are. Good
references are~\cite{bhatia:positivematrices,paulsen:completelypositive}.
The main notion is that of a matrix algebra $\mathbb{M}_n(A)$ over a C*-algebra
$A$. Its elements are $n$-by-$n$ matrices with entries in $A$, given
C*-structure by 
$(a_{ij}) \cdot (b_{ij}) = \sum_{k=1}^n a_{ik} b_{kj}$, 
$(a_{ij})^* = (a_{ji}^*)$, and $\|(a_{ij})\| = \sup \{ \|(a_{ij}) x\| \mid x \in
A^n, \|x\|=1\}$, where $\|(x_1,\ldots,x_n)\|^2 = \sum_{i=1}^n \|x_1\|^2$.
The following well-known lemma then follows directly.

\begin{lemma}
  If $A$ is a finite-dimensional C*-algebra, then so is $\mathbb{M}_n(A)$. 
  If $f \colon A \to B$ is a linear map, then so is the function $\mathbb{M}_n(f) \colon \mathbb{M}_n(A) \to \mathbb{M}_n(B)$ that
  sends $(a_{ij})$ to $(f(a_{ij}))$. If $f$ is a $*$-homomorphism, then
  so is $\mathbb{M}_n(f)$.
  \qed
\end{lemma}

\begin{definition}\label{def:cp:cstar}
  A linear map $f \colon A \to B$ between C*-algebras is 
  \emph{positive} when for every $a \in A$ there exists $b\in B$ satisfying
  $f(a^*a)=b^*b$. It is \emph{completely positive} when $\mathbb{M}_n(f)$ is
  positive for every $n\in\mathbb{N}$.
\end{definition}

For us it will be convenient to take another, well-known,
viewpoint. If $A$ and $B$ are finite-dimensional C*-algebras, then so
is the algebraic tensor product $A \otimes B$. 

\begin{lemma}\label{lem:cp:cstar}
  Any finite-dimensional C*-algebra $A$ has a canonical
  $*$-isomorphism $\mathbb{M}_n(A) \cong A \otimes \mathbb{M}_n(\mathbb{C})$. Under this
  correspondence, a linear map $f \colon A \to B$ between C*-algebras
  is completely positive if and only if $f \otimes
  \id[\mathbb{M}_n(\mathbb{C})]$ is positive for every $n \in \mathbb{N}$.
\end{lemma}
\begin{proof}
  Borrowing notation from Lemma~\ref{lem:Mn}, the $*$-isomorphism
  $\mathbb{M}_n(A) \to A \otimes \mathbb{M}_n(\mathbb{C})$ is given by $(a_{ij}) \mapsto
  \sum_{i,j=1}^n a_{ij} \otimes e_{ij}$. One easily verifies that this
  preserves the multiplication and involution. The second statement
  follows directly from Definition~\ref{def:cp:cstar} by unfolding this isomorphism.
\end{proof}

Now we have phrased the concrete notion of complete positivity in
terms of tensor products with matrix algebras (cf.\
Theorem~\ref{thm:pos-elems}), and given that matrix algebras in
$\FHilb$ are precisely the algebras of the form $L(H)$ for some
Hilbert space $H$, we can determine $\CPs[\FHilb]$.

\begin{theorem}\label{thm:cpstarfhilb}
  $\CPs[\FHilb]$ is equivalent to the category of finite-dimensional
  C*-algebras and completely positive maps. 
\end{theorem}
\begin{proof}
  Define a functor $E$ from $\CPs[\FHilb]$ to the category of
  finite-dimensional C*-algebras and completely positive maps, acting
  on objects as in Theorem~\ref{thm:dfa-cstar} and as identity on
  morphisms.
  
  Proposition~\ref{prop:dNFAsinFHilb} and Theorem~\ref{thm:dfa-cstar}
  show that $E$ is surjective on objects. Suppose $f \colon (A,
  \whitemult) \to (B, \graymult)$ is a morphism in $\CPs[\V]$. Then 
  Theorem~\ref{thm:pos-elems} shows that $E(f)$ must be completely
  positive, as characterised by Lemma~\ref{lem:cp:cstar}. Therefore
  $E$ is well-defined. Conversely, any completely positive map $g$
  between C*-algebras satisfies the CP*--condition because of
  Lemma~\ref{lem:cp:cstar}, so $E$ is full and hence an equivalence
  of categories. 
\end{proof}

\begin{remark}\label{rem:Lnotessentiallysurjective}
  It follows from the previous theorem that the embedding $L$
  does not extend to an equivalence of categories because it is not essentially
  surjective. That is, there are finite-dimensional C*-algebras, such
  as $A=\mathbb{M}_1(\C) \oplus \mathbb{M}_2(\C)$, that are not isomorphic to a factor:
  $\dim(A)=1^2+2^2=5\neq n^2=\dim(\mathbb{M}_n(\C))$.
\end{remark}

\section{Sets and relations}\label{sec:Rel}

The previous section justified regarding normalisable dagger
Frobenius algebras as generalised (finite-dimensional) C*-algebras. This section 
considers $\CPs[\V]$ for $\V=\Rel$, the category of sets and relations. Starting
with objects, we immediately see that these `generalised C*-algebras'
are quite different.

\begin{proposition}\label{prop:cpsrelobjs}
  All normalisable dagger Frobenius algebras in $\Rel$ are
  special. Therefore they are (in 1-to-1 correspondence with small) groupoids. 
\end{proposition}
\begin{proof}
  We have to prove that normalisability implies speciality in \Rel;
  for this it suffices to show that $z^2=1$.
  Now, the normaliser $z$ is an isomorphism. In \Rel, this means it is
  (the graph of) a bijection. But $z$ is also positive, and hence
  self-adjoint. This means it is equal to its own inverse. Therefore
  $z^2=1$. The second statement now follows directly
  from~\cite[Theorem~7]{heunencontrerascattaneo:groupoids}.
\end{proof}

Next, we turn to determining the morphisms of $\CPs[\Rel]$.

\begin{definition}
  A relation $R \subseteq \Mor(\cat{G}) \times
  \Mor(\cat{H})$ between groupoids $\cat{G},\cat{H}$ \emph{respects
    inverses} when
  \begin{align}
    gRh \Longleftrightarrow g^{-1} R h^{-1},
    \qquad
    gRh \Longrightarrow \id[\dom(g)] R \id[\dom(h)].\label{eq:homrel}
  \end{align}
\end{definition}

\begin{proposition}\label{prop:cpsrel}
  $\CPs[\Rel]$ is (isomorphic to) the category of groupoids and relations respecting inverses.    
\end{proposition}
\begin{proof}
  In general, a morphism $R \subseteq (X \times X) \times (Y \times
  Y)$ in \Rel is completely positive if and only if  
  \begin{equation}\label{eq:cprel}
    (x',x)R(y',y) \Longleftrightarrow (x,x')R(y,y'), 
    \qquad
    (x',x)R(y',y) \Longrightarrow (x,x)R(y,y).
  \end{equation}
  If $\cat{G}$ and $\cat{H}$ are groupoids, corresponding to Frobenius algebras $(G,\whitemult)$
  and $(H,\graymult)$, and $R \subseteq G \times H$, 
  \begin{align*}
    \whitecoaction & = \{ ((g,g'),g^{-1}\circ g') \in G^3 \mid \cod(g)=\cod(g') \}, \\
    \grayaction \circ R \circ \whitecoaction & = \{ ((g,g'),(h,h')) \in G^2 \times H^2 \mid 
    \cod(g)=\cod(g'), \cod(h)=\cod(h'), 
    (g^{-1} \circ g') R (h^{-1} \circ h') \}.
  \end{align*}
  Substituting this into~\eqref{eq:cprel} translates precisely
  into~\eqref{eq:homrel}. 
\end{proof}



We close this section by investigating the embedding $L \colon
\CPM[\Rel] \to \CPs[\Rel]$. Recall that a category is
\emph{indiscrete} when there is precisely one morphism between each
two objects. Indiscrete categories are automatically groupoids. 

\begin{lemma}\label{lem:indiscrete}
  The essential image of the embedding $L \colon \CPM[\Rel] \to
  \CPs[\Rel]$ is the full subcategory of $\CPs[\Rel]$
  consisting of indiscrete (small) groupoids.
\end{lemma}
\begin{proof}
  Let $X$ be an object in $\CPM[\Rel]$, that is, a set.
  By definition, $L(X)$ corresponds to a groupoid with set of morphisms $X \times X$, and composition
  \[
    (y_1,y_2) \circ (x_1,x_2) = \left\{ \begin{array}{ll} (y_1,x_2) & \text{ if } y_2=x1, \\ \text{undefined} & \text{ otherwise}. \end{array} \right.
  \]
  We deduce that the objects of $L(X)$ correspond to identities,
  i.e.\ pairs $(x_1,x_2)$ with $x_1=x_2$. So objects of $L(X)$
  just correspond to elements of $X$. Similarly, we find that
  $\dom(x_1,x_2)=x_2$ and $\cod(x_1,x_2)=x_1$. Hence $(x_1,x_2)$ is a
  morphism $x_2 \to x_1$ in $L(X)$, and it is the unique such. 
\end{proof}


\section{Splitting idempotents}\label{sec:idempotents}

This section compares the CP*--construction to Selinger's second
solution to the problem of classical channels. First, recall this
construction of splitting dagger idempotents~\cite{selinger:idempotents}.

\begin{definition}
  Let $\V$ be a dagger category. The category $\Split[\V]$ has as
  objects $(A,p)$ where $p \colon A \to A$ is a morphism in $\V$
  satisfying $p \circ p = p = p^\dag$; its morphisms $(A,p) \to (B,q)$ are
  morphisms $f \colon A \to B$ in $\V$ satisfying $f = q \circ f \circ p$.
\end{definition}

If $\V$ is a dagger compact category, then so is
$\Split[\V]$ (see~\cite[Proposition~3.16]{selinger:idempotents}). 
We will need the following assumption, that is satisfied in both
$\Rel$ and $\FHilb$.

\begin{definition}\label{def:alg-sq-roots}
  A dagger compact category $\V$ is said to \emph{have algebraic
    square roots} when, given any normalisable Frobenius algebra on
  $A$ and any central positive definite morphism $f \colon A \to A$,
  there exists a central morphism $g \colon A \to A$
  such that $f=g \circ g$.
\end{definition}

We thank the anonymous referee for pointing us towards the following
proposition. 

\begin{proposition}\label{prop:split}
  Let $\V$ be a dagger compact category that has algebraic square roots.
  There is a functor $F \colon \CPs[\V] \to \Split[\CPM[\V]]$,
    acting as $F(A,\whitemult,\whitenorm) = \whiteaction \circ 
      \whitecoaction \circ (\whiteconorm \otimes \whitenorm)$ on
      objects, and as $F(f)= \grayaction \circ \graynorm \circ f \circ \whitenorm
      \circ \whitecoaction$ on morphisms. It is full, faithful,
      and strongly dagger symmetric monoidal. 
\end{proposition}
\begin{proof}
  First, notice that $p=F(A,\whitemult,\whitenorm)$ is indeed a
  well-defined object of $\Split[\CPM[\V]]$: clearly $p=\whiteaction
  \circ \whitenorm \circ \whitenorm \circ \whitecoaction=p^\dag$ by
  centrality of the normaliser, $p \circ p = p$ follows from 
  Lemma~\ref{lem:norm-alt}, and $p$ is completely positive by
  Lemma~\ref{lem:frob-actions}. 
  The assumption of algebraic square roots guarantees that $F(f)$ is
  indeed a well-defined morphism in $\CPM[\V]$; by
  Lemma~\ref{lem:norm-alt}, it is in fact a well-defined morphism in
  $\Split[\CPM[\V]]$. 
  Moreover, it is easy to see that an arbitrary $\V$-morphism
  $h \colon A^* \otimes A \to B^* \otimes B$ is a well-defined
  morphism in $\CPs[\V]$ if and only if it is a well-defined morphism
  in $\Split[\CPM[\V]]$.
  \ctikzfig{cpstar_into_split}  
  Both $\CPs[\V]$ and $\Split[\CPM[\V]]$ inherit composition,
  identities, and daggers from $\V$, so $F$ is a full and faithful
  functor preserving daggers. The
  symmetric monoidal structure in both $\CPs[\V]$ and
  $\Split[\CPM[\V]]$ is similarly defined in terms of that of $\V$, making
  $F$ strongly symmetric monoidal. 
\end{proof}


Thus, when $\V$ has algebraic square roots, $\CPs[\V]$ is equivalent
to a full subcategory of $\Split[\CPM[\V]]$: the one obtained by
splitting only the idempotents of the form $F(A,\whitemult,\whitenorm)$. A variation
of~\cite[Proposition~3.16]{selinger:idempotents} shows that splitting
any family of dagger idempotents that is closed under tensor gives a dagger
compact category, and Theorem~\ref{thm:cpstar-smc} follows. 
The proof that we have written out does not need the assumption of
algebraic square roots; moreover, it more explicitly exhibits the
structure of $\CPs[\V]$ as a category of algebras.


In summary, in sufficiently nice cases, the CP*--construction fits between the
CPM--construction and its idempotent splitting.
\[\xymatrix{
  \CPM[\V] \ar^-{L}[r] & \CPs[\V] \ar^-{F}[r] & \Split[\CPM[\V]]
}\]
Both functors are strongly dagger symmetric monoidal, as well as full
and faithful. 
Moreover, their composition is naturally isomorphic to the canonical
inclusion $\CPM[\V] \to \Split[\CPM[\V]]$. 
However, the image of the left functor does not include
classical channels. Similarly, a priori there is no reason why the
right functor should be an equivalence. 
In particular, the middle category seems to capture the right amount
of objects, and provides a constructive way to access them.

\section{Future work}\label{sec:future}

Having an abstract notion of C*-algebra, (and, by extension, an
abstract categorical construction placing classical and quantum
information on equal footing) opens up many avenues for exploration.  
\begin{itemize}
\item Quantum mechanics can be characterised in information-theoretic
  terms~\cite{cliftonbubhalvorson}, but this argument is often
  criticised because it assumes a C*-algebraic framework from the start. 
  The CP*--construction can investigate to what extent this criticism
  is valid and improve on those foundations.
\item We can now abstractly study all sorts of notions from the C*-algebraic
  formulation of quantum information
  theory~\cite{keyl:quantuminformation,keylwerner:lectures}. For
  example, notions of complementarity can be translated between
  abstract and concrete C*-algebras~\cite{coeckeduncan:redgreen,heunen:complementarity}. 
\item The category $\cat{Stab}$ of stabiliser quantum mechanics embeds
  into $\CPs[\Rel]$, opening the door to abstract considerations of
  classical simulable circuits.
\item It would be good to see whether algebraic square roots
    are necessary for Proposition~\ref{prop:split}. We also expect to
    find an example showing that $F$ is not an equivalence. 
\item One could characterise categories of the form $\CPs[\V]$,
  perhaps using environments~\cite{coecke:mixed,coeckeperdrix:environment,coeckeheunen:cp}. 
\item It is worth investigating to what extent our construction 
  generalises to infinite dimension~\cite{abramskyheunen:hstar,coeckeheunen:cp}.
\item On the theoretical side, the CP*--construction might be (lax) monadic.
\item Our construction seems related in spirit to~\cite{vicary:higher}; it would be
  good to make connections precise.
\end{itemize}

\bibliographystyle{eptcs}
\bibliography{cpstar}
\end{document}